\documentclass[journal]{IEEEtran}

\usepackage{algorithm}
\usepackage{algorithmic}
\usepackage{cite}
\usepackage{amsmath,amssymb,amsfonts}
\usepackage{graphicx}
\usepackage{textcomp}
\usepackage{xcolor}
\usepackage{subfigure}
\usepackage{bm}
\usepackage{amsthm}
\usepackage{xurl}
\usepackage{hyperref}

\usepackage[subtle,tracking=normal]{savetrees}

\theoremstyle{definition}
\newtheorem{theorem}{\textbf{Theorem}}

\newtheorem{proposition}{\textbf{Proposition}}

\allowdisplaybreaks[4]
\def\BibTeX{{\rm B\kern-.05em{\sc i\kern-.025em b}\kern-.08em
    T\kern-.1667em\lower.7ex\hbox{E}\kern-.125emX}}

\begin{document}
\title{Fairness as an Investment: Dynamic Participation and Long-Run Profit in Virtual Power Plants}

\author{Liudong~Chen,~\IEEEmembership{Student~Member,~IEEE,}
        and~Bolun~Xu,~\IEEEmembership{Member,~IEEE}
\thanks{L. Chen and B. Xu are with the Department
of Earth and Environmental Engineering, Columbia University, New York, NY, 10027 USA e-mail: \{lc3671,~bx2177@columbia.edu\}.}

}

\maketitle

\begin{abstract}
We show that incorporating fairness constraints into virtual power plant (VPP) operations can incentivize consumer participation and thus improve the aggregator’s long-run profitability. VPPs rely on sustained participation from heterogeneous consumers to provide a variety of grid services whose timing and frequency are often uncertain. As a result, consumers' willingness and ability to provide flexibility evolve over time, creating a dynamic link between past participation and future resource availability. We develop a dynamic aggregation framework to study how fairness in service allocation affects future participation and long-run profitability. By linking current dispatch decisions to future resource availability, we show that fairer allocations can strengthen consumer engagement, expand aggregate availability, and create additional value during high-price and high-demand events. To balance fairness and operational efficiency, we introduce a slack-augmented allocation mechanism that preserves most of the participation benefits from fairness while avoiding unnecessary reductions in service procurement. We derive conditions under which the resulting availability gains outweigh the short-run cost of redistribution and validate the approach using real-world consumer behavior and electricity market data from Norway. 
\end{abstract}


\IEEEpeerreviewmaketitle

\section{Introduction}
Virtual power plants (VPPs) are rapidly emerging in practice to aggregate distributed energy resources (DERs) for coordinated participation in energy markets. In the United States, there is approximately 30~GW of VPP capacity that can shave about 3.75\% of peak demand at a cost of \$43/kW-year, which is 37.7\% and 56.6\% lower than utility-scale batteries and gas peaker plants, respectively~\cite{DOE_VPP}. A rapidly growing class of VPPs is battery-backed, exemplified by Tesla~\cite{Tesla_VPP}, Base Power~\cite{BasePower1}, and Every Electric~\cite{EveryElectric}. VPP models are also expanding to electric vehicles (EVs) and smart-home devices through platforms such as WeaveGrid and EnergyHub~\cite{WeaveGrid, EnergyHub}.

VPP aggregators coordinate enrolled consumers by offering incentive payments for service provision, which is represented by energy delivery that supports wholesale market participation, distribution-level congestion relief, feeder peak shaving, demand response, and other flexibility services. However, consumers differ substantially in their ability to provide service due to heterogeneous DER ownership, operating preferences, and willingness to participate. For example, higher-income households often own larger DER portfolios, such as multiple EVs or larger batteries~\cite{flexibility2}, and may be willing to supply services at lower cost. As a result, an aggregator that maximizes short-run profit tends to repeatedly dispatch more flexible consumers, raising fairness concerns about unequal participation and compensation. These concerns are increasingly reflected in policy and program practice. The U.S. Department of Energy has emphasized affordable VPP development and broader access to participation~\cite{DOE_VPP, DOE_VPP2}. Pacific Gas and Electric in California has implemented a ``first-of-its-kind'' VPP initiative requiring that at least 60\% of participants come from disadvantaged or low-income communities~\cite{PGE_SAVE_VPP}. EnergyHub similarly reports that aligning program design with user behavioral trends and equity goals can increase EV-driver enrollment~\cite{Energyhub_report}. Tesla and Base Power also emphasize meaningful compensation and maintaining a minimum level of backup capability for participating households~\cite{Tesla_VPP, BasePower1}.

% Although fairness considerations are increasingly important in practice, 
The existing literature mainly studies fairness in single-period settings and treats fairness as a trade-off against profit. In contrast, we ask: \textit{Can fairness instead serve as an investment that increases long-run VPP profitability?} We answer in the affirmative: when participation responds endogenously to recent activity, fairness can expand future aggregate availability, which is subsequently monetized during extreme market conditions.
Our main contributions are:
\begin{itemize}
\item We develop a multi-period VPP aggregation framework in which consumer availability evolves with past participation, and formulate a non-anticipatory operating problem with dispatch fairness constraints.

\item We show that fairness affects profits through two channels: reallocation across consumers that increases future availability under an increasing and strictly concave state transition; and curtailment that imposes only costs.

\item We propose a slack-augmented fairness design that eliminates curtailment while preserving participation benefits, and prove that it weakly dominates strict fairness in realized profit.

\item We derive market price and energy requirement conditions under which the availability gains from fairness outweigh its short-run costs.

\item Using real-world data from Norway, we demonstrate that appropriately chosen fairness levels can improve long-run VPP profitability.

\end{itemize}
The paper is organized as follows. Section~II reviews the background literature, and Section~III presents the formulation. Section~IV analyzes the fairness constraint, and Section~V presents the slack-augmented fairness design and its theoretical results. Section~VI presents the case study, and Section~VII is the conclusion.

\section{Literature Review}
\subsection{VPP business models in practice}\label{sec:vpp_practice}
Commercial VPPs typically operate through centralized dispatch of enrolled DER assets, issue dispatch instructions directly, and compensate participants for the service provided. In battery-backed programs, Tesla monitors and dispatches enrolled Powerwall batteries and compensates participants at \$2--\$5 per kWh discharged, substantially above retail electricity rates~\cite{Tesla_VPP}. Base Power leases batteries directly to households and operates them centrally~\cite{BasePower, BasePower1}. Every Electric similarly installs residential batteries with no upfront cost and recovers costs through demand-response revenue, sharing the resulting margin with participating consumers~\cite{EveryElectric}. Beyond batteries, platforms such as WeaveGrid and EnergyHub coordinate EV chargers, smart thermostats, batteries, and other flexible DERs through asset-level integrations~\cite{WeaveGrid, EnergyHub}, while Pacific Gas and Electric's SmartAC program directly cycles enrolled thermostats with consumer consent~\cite{PGE_SmartAC}. 

A common feature across these models is that participation is voluntary and must be sustained over time. Consumers may opt out or leave if the program does not provide sufficient value; for example, in Every Electric's model, infrequent dispatch reduces compensation opportunities and weakens the value of keeping the battery. This highlights the link between recurring payments, consumer retention, and sustained resource availability.
Industry reports on scaling VPPs emphasize that recurring compensation and frequent engagement are central to sustaining participation~\cite{LBNL_scaling, LBNL_DER_utility, NESO_DSR}. Broader socioeconomic evidence also shows that consumer retention depends on past participation~\cite{bies2021push}; and in energy demand response, repeated economic incentives can produce persistent responses that outlast experimental treatments~\cite{ito2018moral}. Studies of high-DER systems further emphasize that flexibility must be reliably available, not merely nominally enrolled, to create system-level value~\cite{Brown2021DERaggregation}.

\subsection{Fairness in VPP and dynamic allocation}
The existing VPP-fairness literature has incorporated fairness through several criteria, including Rawlsian min-max objectives~\cite{aggregator_consider}, Gini-based or variance-based indices~\cite{define_a, jain1984fairness}, and individual-level fairness notions that compare system outcomes across different fairness criteria and enforcement levels~\cite{ourwork}. These frameworks provide useful ways to measure and enforce fairness, but they typically treat fairness as a static or one-shot constraint within each operating period. As a result, fairness primarily appears as a restriction on the aggregator's feasible set, thereby reducing short-run profit or business competitiveness. What is missing is the closed-loop interaction emphasized in the practical VPP business model: participation is voluntary, payments are recurring, and consumers' future availability can depend on how frequently and fairly they are dispatched.

Fairness in dynamic allocation has been studied more broadly in operations research and computer science. Across-agent fairness balances outcomes among heterogeneous users~\cite{ghodsi2011dominant}; across-time fairness enforces fairness over a horizon and can reduce its welfare cost~\cite{lodi2024fairness}. Restless-bandit fairness limits the time since an arm was last served~\cite{li2022efficient}, and other formulations evaluate fairness only at the terminal stage~\cite{blum2022multi}. Methodologically, these works characterize optimal policies under fairness constraints~\cite{balakrishnan2022scales} or derive online algorithms with regret bounds~\cite{si2022enabling}. Our contribution is different: \emph{we study how fairness affects the evolution of consumer availability itself}. This dynamic channel can reverse the usual conclusion that fairness only reduces competitiveness, since fairer dispatch may sustain future participation and expand the aggregator's profitable resource pool.

\begin{figure}
    \centering
    \includegraphics[width=0.99\linewidth]{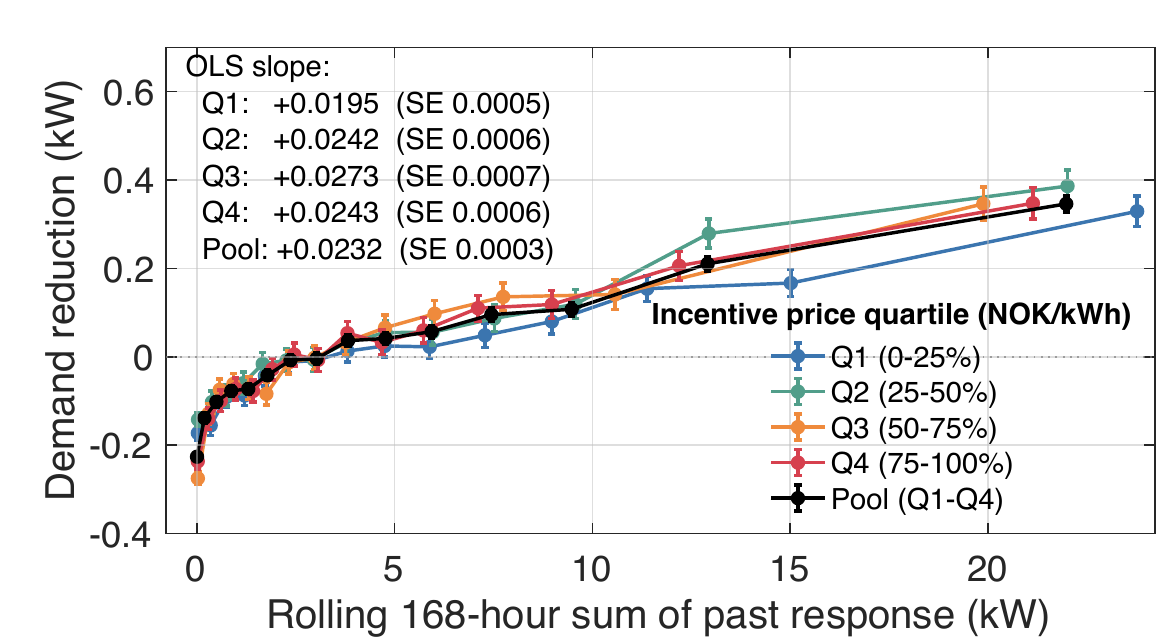}
    \caption{Empirical evidence of participation reinforcement from the Norway iFlex experiment. Households with greater recent response history exhibit larger demand reductions during incentive events, with diminishing marginal effects, consistent with an increasing and concave participation-state transition. OLS is ordinary least squares, and SE is the standard error of the OLS slope.}
    \label{fig:empirical}
\end{figure}

\section{Model Formulation}
\subsection{Oracle problem}

We first state the \emph{Oracle} problem the aggregator would solve if it could observe all consumer and market realizations across the horizon. We adopt a \emph{centralized-dispatch} convention consistent with commercial VPP programs. Let $\mathcal T:=\{1,\dots,T\}$ denote the set of periods. At each $t\in\mathcal T$, the VPP faces at most one service request with a maximum aggregate energy requirement $\bar D_t$ and a compensation rate $\pi_t$. The process ${(\pi_t,\bar D_t)}_{t=1}^T$ is defined on a probability space $(\Omega,\mathcal F,\mathbb P)$ with filtration ${\mathcal F_t}$. In many settings, high-energy requirement events are associated with greater system scarcity and therefore higher compensation. 

Each consumer $i\in\mathcal N$ provides energy $D_{i,t}\in[0,A_{i,t}]$ in period $t$. The \emph{availability} $A_{i,t}\in[0,\bar A]$ is the maximum energy consumer $i$ can provide during the dispatch period, reflecting both physical DER constraints and participation willingness. For battery-backed VPPs, this is primarily dependent on the state of the stored energy reserve. Since real response periods range from $15$ minutes to several hours, $A_{i,t}$ captures both instantaneous power capability and energy sustainability, i.e., the total energy the consumer can contribute without violating constraints. We assume availability is bounded above by $\bar A$ for all $i,t$.

As we discussed in Section~\ref{sec:vpp_practice}, more frequent participation reinforces consumer engagement in practical VPP programs.
Fig.~\ref{fig:empirical} shows that in the iFlex field experiment in Norway, household response is increasing and concave in cumulative recent exposure. We thus motivate the following availability evolution dynamics
\begin{subequations}\label{eq:oracle_dynamics}
\begin{align}
S_{i,t+1} &= \beta S_{i,t} + \rho D_{i,t}, \label{eq:state_update}\\
A_{i,t+1} &= g(S_{i,t+1}, \sigma_t),  \label{eq:availability_map}
\end{align}
\end{subequations}
where $S_{i,t}$ is a latent participation state, $\beta\in(0,1)$ is the persistence of past participation, $\rho>0$ is the marginal effect of current dispatch on the next-period state, and $\sigma_t$ is a random shock realized after $t$ that affects $A_{i,t+1}$ through the increasing map $g$ (with $g(S,\sigma)\le\bar A$ for all $S\in[0,\bar S]$ and $\sigma$), where $\bar S$ is an upper bound on the state space. Initial states are normalized: $S_{1,1}=\cdots=S_{N,1}=:S_1$. The linear update provides a tractable state representation; the nonlinear, bounded map $g$ captures saturation.

Each consumer $i$ has a marginal cost of providing energy $c_i>0$, capturing both (i) the per-unit compensation required to induce participation and (ii) the effective network cost of dispatching consumer $i$, since consumers located farther from the point of connection may create larger distribution losses. We adopt a linear cost specification,
\begin{equation}\label{eq:linear_cost}
C_i(D_{i,t}) = c_i D_{i,t},
\quad
\pi_t>c_i,\ \forall i,t,
\end{equation}
consistent with battery-backed VPP operations, where degradation, the opportunity cost of stored energy, and backup-reserve loss are all incurred per unit of energy delivered. Under linear compensation contracts, modeling participation reinforcement through dispatch is equivalent to modeling it through recurring compensation exposure. 

The \textbf{Oracle} problem maximizing the long-run profit is defined as
\begin{subequations}\label{eq:oracle}
\begin{align}
\max_{\{D_{i,t}\}}\ &\mathbb E \left[\sum_{t=1}^T \sum_{i\in \mathcal N} \big(\pi_t - c_i\big) D_{i,t}\right] \label{eq:oracle_obj}\\
\text{s.t.}&\ (\forall t\in\mathcal T) \nonumber\\
& 0\le D_{i,t}\le A_{i,t}\, , \  \forall i\in\mathcal N\label{eq:oracle_box}\\
& \sum_{i\in\mathcal N} D_{i,t} \le \bar D_t,  \label{eq:oracle_aggregate}\\
& A_{i,t+1} = g(\beta S_{i,t}+\rho D_{i,t},\sigma_t)\, , \  \forall i\in\mathcal N \label{eq:oracle_state}
\end{align}
\end{subequations}
The formulation reveals two channels of profit. The \emph{contemporaneous} channel: each unit dispatched in period $t$ earns marginal profit $\pi_t-c_i$ via~\eqref{eq:oracle_obj}. The \emph{intertemporal} channel: current dispatch $D_{i,t}$ feeds into the next-period state $S_{i,t+1}$ via~\eqref{eq:oracle_state}, expanding future availability and enlarging the feasible set~\eqref{eq:oracle_box} for future profitable sales. An optimal Oracle operation must thus balance these channels intertemporally.

\subsection{Non-anticipatory operation with fairness constraints}
The Oracle problem is not implementable in practice because the aggregator lacks access to the consumer-availability evolution model \eqref{eq:oracle_dynamics}, which is rooted in private consumer behaviors. Without this constraint, the VPP aggregation problem loses its inter-temporal structure and becomes a \textbf{Greedy operation} that only maximizes profit from the current time period based on observed $(\pi_t,\bar D_t)$:
\begin{align}
\max_{D_{t}\in \mathbb R^N}\ & \Pi_t=\sum_{i\in\mathcal N}(\pi_t-c_i)D_{i,t},
\text{ s.t. \eqref{eq:oracle_box} and \eqref{eq:oracle_aggregate} } \label{eq:operating}
\end{align}
Where, although the states still evolve according to~\eqref{eq:oracle_dynamics}, they are not modeled by the aggregator. A natural solution in this case is rank-based dispatch: order consumers by $c_i$ and dispatch the lowest-cost consumers first until the energy requirement or availability cap binds. This myopic policy captures the contemporaneous channel but ignores the intertemporal one. 

A practical way to approximate this intertemporal diversification benefit is to impose a fairness constraint that limits repeated concentration on low-cost consumers and thereby expands future aggregate availability. We therefore introduce a \emph{dispatch fairness constraint} that bounds per-period dispersion in energy provision, motivated by equitable-access requirements in EV charging~\cite{demand_ev} and minimum-access guarantees in emergency load shedding~\cite{demand_Shedding}. Let $D_t^0$ denote an optimizer of the Greedy operation~\eqref{eq:operating}. We define a fairness level $\alpha\in[0,1]$ and the benchmark dispersion 
\begin{equation}
    \Delta_t:=\max_{i\in\mathcal N}D_{i,t}^0-\min_{i\in\mathcal N}D_{i,t}^0\ge0\, ,\ D_{i,t}^0\in\arg\eqref{eq:operating} .\label{eq:benchmark_dis}
\end{equation}
When $\Delta_t=0$, the benchmark already dispatches every consumer equally, so the fairness constraint~\eqref{con:fairness} is automatically satisfied and carries no content; we therefore restrict attention to periods with $\Delta_t>0$. Then the fair operating problem augments~\eqref{eq:operating} with
\begin{equation}
    \max_{i\in\mathcal N}D_{i,t}^\alpha-\min_{i\in\mathcal N}D_{i,t}^\alpha
    \le (1-\alpha)\Delta_t.
    \label{con:fairness}
\end{equation}
where $D_{i,t}^\alpha$ represents the fairness-augmented dispatch.
When $\alpha=0$, the constraint recovers the profit-only benchmark from \eqref{eq:operating}; when $\alpha=1$, it enforces the most equal feasible dispatch. 

The \textbf{fairness-constrained} aggregation problem is thus defined as 
\begin{align}
\max_{D_{t}^\alpha\in \mathbb R^N}\ & \Pi^\alpha_t=\sum_{i\in\mathcal N}(\pi_t-c_i)D_{i,t}^\alpha,
\text{ s.t. \eqref{eq:oracle_box}, \eqref{eq:oracle_aggregate}, \eqref{eq:benchmark_dis}, \eqref{con:fairness} } 
\end{align}
The resulting non-anticipatory long-run profit is $\sum_{t=1}^T \Pi_t^\alpha$. Because the fairness constraint is defined relative to the benchmark dispatch, the difference $D_t^\alpha-D_t^0$ at a given availability state and market realization isolates the effect of fairness itself, independent of subsequent state evolution. Throughout, we use the non-anticipatory profit-maximizing policy as the benchmark and interpret fairness as an implementable proxy for the intertemporal value that a fully informed Oracle could capture. 

\section{Analysis of the Fairness Constraint}
We analyze the fairness impact on the operating model and show that fairness incurs costs when it induces curtailment of aggregate energy collection, while reallocation across consumers raises aggregate availability.

\subsection{Fairness results in service curtailment}
We first show that, compared to the Greedy operation, the fairness-constrained operation may reduce aggregate energy collection. By restricting dispatch dispersion through \eqref{con:fairness}, the fairness constraint shifts collection toward higher-cost consumers. When the availability of these consumers is insufficient to absorb the reallocated service, the optimization may be forced to reduce total collection, resulting in \emph{curtailment}. More generally, the trivial allocation $D_{i,t}=0$ for all $i$ satisfies \eqref{con:fairness} exactly, demonstrating that perfect equality can be achieved by foregoing service altogether. The next proposition characterizes the conditions under which fairness induces aggregate energy curtailment.

\begin{proposition}[Condition for fairness-induced curtailment]\label{lem:curtail_cause}
Let $\underline A_t:=\min_{i\in\mathcal N}A_{i,t}$. The fairness constraint~\eqref{con:fairness} introduces curtailment $R_t(\alpha):=\sum_{i\in\mathcal N}D_{i,t}^0-\sum_{i\in\mathcal N}D_{i,t}^\alpha>0$ whenever
\begin{equation}
\sum_{i\in\mathcal N}D_{i,t}^0 > \sum_{i\in\mathcal N}\min\bigl\{A_{i,t},\;\underline A_t+(1-\alpha)\Delta_t\bigr\}.
\label{eq:curtail_condition}
\end{equation}
\end{proposition}
\begin{proof}
The range cap~\eqref{con:fairness} forces $\max_i D_{i,t}^\alpha-\min_i D_{i,t}^\alpha\le(1-\alpha)\Delta_t$, and the box constraint gives $\min_i D_{i,t}^\alpha\le\min_i A_{i,t}=\underline A_t$. Hence every $D_{i,t}^\alpha\le\min_i D_{i,t}^\alpha+(1-\alpha)\Delta_t\le\underline A_t+(1-\alpha)\Delta_t$, and combining with the box constraint, $D_{i,t}^\alpha\le\min\{A_{i,t},\,\underline A_t+(1-\alpha)\Delta_t\}$. Summing over $i$ gives $\sum_i D_{i,t}^\alpha\le\sum_i\min\{A_{i,t},\,\underline A_t+(1-\alpha)\Delta_t\}$, so whenever the right-hand side is below $\sum_i D_{i,t}^0$, the curtailment $R_t(\alpha)$ is strictly positive.
\end{proof}

Proposition~\ref{lem:curtail_cause} shows that the fairness constraint mechanically reduces total collection: the individual availability caps $A_{i,t}$ on the high-cost side prevent enough mass from being lifted to keep the fair aggregate $\sum_{i\in\mathcal N}D_{i,t}^\alpha$ at the benchmark level. We thus decompose the change $D_t^0\to D_t^\alpha$ into an aggregate \emph{curtailment} $R_t(\alpha)$ and a within-period \emph{reallocation}
\begin{equation}
M_t(\alpha):=\sum_{i\in\mathcal N}\big(D_{i,t}^\alpha-D_{i,t}^0\big)_+\ge 0,
\label{eq:relocation}
\end{equation}
which preserves the total collection by construction. Here $(x)_+ = \max\{0,x\}$.

\subsection{The benefit of reallocation from fairness}
The decomposition $D_t^0\to D_t^\alpha$ into reallocation $M_t(\alpha)$ and curtailment $R_t(\alpha)$ separates two physically distinct effects of fairness. We first analyze the effect of reallocation by considering cases where $R_t(\alpha)=0$ in every period, so the change $D^0_t\to D^\alpha_t$ is a pure reallocation $M_t(\alpha)$. The following proposition gives the structural condition under which fairness raises expected next-period aggregate availability.

\begin{proposition}[Fairness induces higher total availability]
\label{prop:availability_gain}
Suppose $\sum_{i\in \mathcal N}D_{i,t}^\alpha=\sum_{i\in \mathcal N}D_{i,t}^0$ and $\alpha >0$. Under the state transition~\eqref{eq:oracle_dynamics}, if the transition function $g$ is strictly increasing and strictly concave on $[0,\bar S]$, then
\begin{equation}
    \sum_{i\in \mathcal N}\mathbb E[A_{i,t+1}^\alpha\mid \mathcal F_t] > \sum_{i\in \mathcal N}\mathbb E[A_{i,t+1}^0\mid \mathcal F_t].
\end{equation}
If $g$ is affine, then the two totals are equal.
\end{proposition}

Proposition~\ref{prop:availability_gain} shows that, when $g$ is increasing and concave, fairness increases expected aggregate availability in the next period, thereby recovering part of the intertemporal value captured by the Oracle problem. The monotonicity assumption reflects that a higher participation state should not decrease future availability, while concavity captures diminishing returns: additional participation has a larger effect when the participation state is low than when it is already high. By limiting concentration, fairness reallocates participation toward high-cost consumers with lower states, where the marginal gain in availability is greater. This condition is also practically reasonable. In economics, the same improvement in an attribute often yields a larger marginal gain when the initial level is low~\cite{horowitz2007test}, and empirical evidence from demand response documents a similar saturation pattern, as reflected in Fig.~\ref{fig:empirical} and related studies~\cite{sarran2021data, chen2023saturation}.

\subsection{The cost of curtailment from fairness}\label{sec:cost_fairness}
Although reallocation incurs a contemporaneous cost by shifting collection from low-cost to high-cost consumers, it sustains higher future availability via Proposition~\ref{prop:availability_gain}. Curtailment offers no such intertemporal offset. The next proposition shows that curtailment imposes two distinct costs: a contemporaneous loss from foregone profitable sales, and an intertemporal loss from a reduced next-period availability state.

\begin{proposition}[Curtailment imposes only costs]\label{prop:curtail_cost}
Suppose $g$ is strictly increasing on $[0,\bar S]$. Consider any nonnegative curtailment profile $(r_i)_{i\in\mathcal N}$ with $r_i\ge 0$ and $\sum_i r_i=R>0$, applied to a dispatch $D_t$ to produce $\tilde D_{i,t}:=D_{i,t}-r_i\ge 0$. Then:
\begin{enumerate}
    \item (Contemporaneous profit loss.) $\sum_i(\pi_t-c_i)\tilde D_{i,t} < \sum_i(\pi_t-c_i)D_{i,t}$.
    \item (Intertemporal availability loss.) $\sum_i\mathbb E[\tilde A_{i,t+1}\mid\mathcal F_t] < \sum_i\mathbb E[A_{i,t+1}\mid\mathcal F_t]$.
\end{enumerate}
where $\tilde A_{i,t+1}=g(\beta S_{i,t}+\rho\tilde D_{i,t},\sigma_t)$ and $A_{i,t+1}=g(\beta S_{i,t}+\rho D_{i,t},\sigma_t)$.
\end{proposition}
\begin{proof}
\emph{Part (1).} The difference of contemporaneous profit with $\tilde D_{i,t}$ and $D_{i,t}$ is calculated by $\sum_i(\pi_t-c_i)(\tilde D_{i,t}-D_{i,t})=-\sum_i(\pi_t-c_i)r_i <  0$, since $\pi_t-c_i>0$ for all $i,t$.

\emph{Part (2).} Under~\eqref{eq:oracle_dynamics}, $\tilde S_{i,t+1}-S_{i,t+1}=\beta S_{i,t}+\rho \tilde D_{i,t} - (\beta S_{i,t}+\rho D_{i,t}) = -\rho r_i\le 0$. Since $g$ is increasing in $S$, $\mathbb E[\tilde A_{i,t+1}\mid\mathcal F_t] = \mathbb E_\sigma[g(\tilde S_{i,t+1},\sigma)] \leq \mathbb E_\sigma[g(S_{i,t+1},\sigma)] = \mathbb E[A_{i,t+1}\mid\mathcal F_t]$. Summing over $i$ yields a strict inequality, since $\sum_i r_i=R>0$ and the strict monotonicity of $g$.
\end{proof}

Proposition~\ref{prop:curtail_cost} shows that curtailment is a pure tax on the aggregator: it lowers current profit because the curtailed units would have been sold at a strictly positive margin, and it lowers expected next-period availability through the reduced participation state. This motivates a fairness design that suppresses $R_t(\alpha)$ while preserving the reallocation component $M_t(\alpha)$.

\section{Augmented Fairness Design and Analysis}
In this section, we augment the non-anticipatory operating model with a slack variable in the fairness constraint to reduce curtailment costs, and analyze how slack reshapes the profitability of fairness. 

\subsection{Slack-augmented operating model}
Proposition~\ref{prop:curtail_cost} shows that curtailment always reduces profit, and Proposition~\ref{lem:curtail_cause} characterizes when curtailment arises: the interaction of individual availability caps~\eqref{eq:oracle_box} with the fairness range cap~\eqref{con:fairness}. In principle, the aggregator could avoid curtailment by tuning $\alpha$ ex-ante, but this requires knowledge of consumer preferences and future market realizations that are not available in the practical settings (the same reason it cannot solve the Oracle problem~\eqref{eq:oracle} directly). We therefore augment the operating problem~\eqref{eq:operating} with two coupled modifications: a slack variable $s_t\ge 0$ that relaxes the fairness constraint endogenously when strict fairness would induce curtailment, and a no-curtailment constraint that ties the slack to its intended purpose.

We design the \textbf{slack-augmented fairness} operating problem at period $t$ as
\begin{subequations}\label{eq:operating_slack}
\begin{align}
\max_{D_t\in \mathbb R^N,\,s_t}\ &\sum_{i\in\mathcal N}(\pi_t-c_i)D_{i,t}-\lambda s_t, \label{eq:op_obj_slack}\\
\text{ s.t. } & \text{\eqref{eq:oracle_box} and \eqref{eq:oracle_aggregate} } \nonumber\\
& \max_{i\in\mathcal N}D_{i,t}-\min_{i\in\mathcal N}D_{i,t}\le(1-\alpha)\Delta_t + s_t, \label{con:fairness_slack}\\
& \sum_{i\in\mathcal N}D_{i,t} = \sum_{i\in\mathcal N}D_{i,t}^0, \label{con:nocurtail}\\
& s_t\ge 0, \label{con:slack_nonneg}
\end{align}
\end{subequations}
where \eqref{eq:oracle_box} and \eqref{eq:oracle_aggregate} are the individual-availability and aggregate-service constraints shared by all operating models. Constraint~\eqref{con:nocurtail} requires the aggregate dispatch to match the Greedy benchmark, thereby eliminating curtailment, while~\eqref{con:fairness_slack} allows the fairness cap to be relaxed by $s_t$ whenever necessary to satisfy~\eqref{con:nocurtail}. The penalty term $\lambda s_t$ discourages unnecessary relaxation of the fairness constraint. Let $\mathcal C_1:=\max_i c_i-\min_i c_i$. Under the no-curtailment equality, an additional unit of slack can only improve the unpenalized objective by reallocating one unit from a higher-cost consumer to a lower-cost consumer, whose value is at most $\mathcal C_1$. Thus choosing $\lambda\ge\mathcal C_1$ ensures that slack is activated only to avoid violating \eqref{con:nocurtail}, i.e., when the fairness constraint would otherwise induce curtailment.

Note that the penalty $\lambda s_t$ in~\eqref{eq:op_obj_slack} is a regularizer that shapes the aggregator's behavior; it does not correspond to a payment the aggregator actually incurs. Accordingly, when evaluating realized aggregator profit, we report
\begin{equation}
\Pi_t^{\alpha,\lambda}=\sum_{i\in\mathcal N}(\pi_t-c_i)D_{i,t}^{\alpha,\lambda},
\label{eq:realized_profit}
\end{equation}
which excludes the penalty term, where $D_{i,t}^{\alpha,\lambda}$ represents the optimized dispatch in the slack-augmented fairness operation with fairness $\alpha$ and slack price $\lambda$.

\subsection{Trade-off with strict fairness}\label{sec:dominance}
By Proposition~\ref{prop:curtail_cost}, curtailment always reduces aggregator profit. The slack-augmented model is constructed precisely to eliminate this cost: in periods where strict fairness would induce curtailment, the slack relaxes the fairness cap, thereby preventing the curtailment loss. In periods without curtailment under strict fairness, the slack remains zero, and the two models coincide. The next theorem formalizes this dominance.

\begin{theorem}[Slack dominates strict fairness in realized profit]\label{thm:dominance}
Fix $\alpha\in[0,1]$ and any finite $\lambda\ge\mathcal{C}_1$. Then for every period $t$,
\begin{equation}
\Pi_t^{\alpha,\lambda}\ge\Pi_t^\alpha,
\label{eq:dominance}
\end{equation}
with equality whenever $R_t(\alpha)=0$ and strict inequality whenever $R_t(\alpha)>0$. Consequently, $\sum_{t=1}^T\Pi_t^{\alpha,\lambda}\ge\sum_{t=1}^T\Pi_t^\alpha$, with strict inequality whenever $R_t(\alpha)>0$ for at least one $t$.
\end{theorem}

Theorem~\ref{thm:dominance} shows that the slack-augmented model weakly dominates the model with strict fairness in realized profit at every period, with strict improvement exactly in the periods where strict fairness induces curtailment. The dominance holds for any $\lambda \geq \mathcal{C}_1$, because the penalty only affects which feasible solution the aggregator selects, not whether the no-curtailment outcome is feasible. Since slack variables relax the fairness range cap, the de facto fairness level may be lower than the declared fairness level $\alpha$. The penalty parameter $\lambda$ controls the extent of this relaxation by determining how aggressively the aggregator sacrifices fairness to avoid curtailment. We quantify the resulting de facto fairness levels in the Case Study. 
 % The role of $\lambda$ is instead to shape the de facto fairness, by trading off how aggressively the aggregator relaxes the fairness cap when curtailment threatens.

\subsection{Profitability of slack-augmented fairness}
We evaluate the profitability of the slack-augmented fairness model,
focusing on the requirement for the market conditions. A key structural
advantage of the slack-augmented model is that it operates entirely within
the no-curtailment regime: by~\eqref{con:nocurtail},
$\sum_i D_{i,t}^{\alpha,\lambda}=\sum_i D_{i,t}^0$ in every period, so
$R_t^{\alpha,\lambda}=0$. The slack-augmented adjustment is therefore a
pure reallocation. From Proposition~\ref{prop:availability_gain},
reallocation from fairness increases expected aggregate availability, which
translates into higher long-run profit when future market events monetize
that availability. The next theorem derives an event-level threshold on
the next-event market signal that guarantees that fairness is paid off.

For each event pair $(t,t{+}1)$, define the following
$\mathcal F_{t+1}$-measurable quantities (functions of the realized shock
$\sigma_t$ through the next-event availabilities):
\begin{equation}
\begin{aligned}
&\Delta A_{t+1}:=\sum_i\bigl(A^{\alpha,\lambda}_{i,t+1}-A^0_{i,t+1}\bigr),\quad
\delta_{t+1}:=\bar D_{t+1}-\sum_i A^0_{i,t+1},\\
&C^{\mathrm{re}}_t:=\sum_i c_i\bigl(D^{\alpha,\lambda}_{i,t}-D^0_{i,t}\bigr),\quad
C^{\mathrm{mix}}_t:=\sum_i c_i\bigl(A^{\alpha,\lambda}_{i,t+1}-A^0_{i,t+1}\bigr),\\
&\bar c^{0}_{t+1}:=\frac{\sum_i c_i A^0_{i,t+1}}{\sum_i A^0_{i,t+1}},\quad
c_{\max}:=\max_{i\in\mathcal N}c_i.
\end{aligned}
\label{eq:threshold_quantities}
\end{equation}
$\Delta A_{t+1}$ is the slack-induced availability gain at $t{+}1$, $\delta_{t+1}$ is the energy requirement cap-limited dispatched volume, $C^{\mathrm{re}}_t$ is the in-period reallocation cost, $C^{\mathrm{mix}}_t$ is the cost of the availability-mix change, and $\bar c^{0}_{t+1}$, $c_{\max}$ are the benchmark availability-weighted average and maximum consumer costs.

\begin{theorem}[Event-level price and quantity threshold]
\label{thm:threshold_slack}
Fix $\alpha\in(0,1]$, $\lambda\ge\mathcal C_1$, and $g$ strictly concave
and strictly increasing on $[0,\bar S]$. The slack-augmented design at
level $\alpha$ contributes non-negatively to realized profit on the
event pair $(t,t{+}1)$ under either of the following conditions:
\begin{itemize}
\item Case 1 (supply-binding): if and only if
\begin{equation}
\bar D_{t+1}\ge\sum_i A^{\alpha,\lambda}_{i,t+1},\quad
\pi_{t+1}\ge\frac{C^{\mathrm{re}}_t+C^{\mathrm{mix}}_t}{\Delta A_{t+1}};
\label{eq:threshold_case1}
\end{equation}
\item Case 2 (energy requirement-binding): if
\begin{align}
&\sum_i A^0_{i,t+1} <\bar D_{t+1}<\sum_i A^{\alpha,\lambda}_{i,t+1},\nonumber\\
&\pi_{t+1}\ge c_{\max}+\frac{C^{\mathrm{re}}_t+(c_{\max}-\bar c^{0}_{t+1})\sum_i A^0_{i,t+1}}{\delta_{t+1}}.
\label{eq:threshold_case2}
\end{align}
\end{itemize}
\end{theorem}

Theorem~\ref{thm:threshold_slack} shows fairness operates as an \emph{investment}: a current-period reallocation loss $C^{\mathrm{re}}_t$ buys a contingent claim on the next event's joint price-energy requirement realization $(\pi_{t+1},\bar D_{t+1})$, recovering the intertemporal value that the Oracle problem captures with full foresight in a contract-friendly, implementable form.

In Case 1 the full $\Delta A_{t+1}$ clears and the threshold is necessary and sufficient; in Case 2 energy requirement caps slack at $\bar D_{t+1}$, only $\delta_{t+1}<\Delta A_{t+1}$ extra units clear, and the threshold rises strictly. Writing $a:=C^{\mathrm{re}}_t$ and $b:=(c_{\max}-\bar c^{0}_{t+1})\sum_i A^0_{i,t+1}\ge 0$ (both non-negative), the right-hand-side (RHS) in Case 2 minus the RHS in Case 1 gives
\begin{equation}
\frac{\sum_i(c_{\max}-c_i)A^{\alpha,\lambda}_{i,t+1}}{\Delta A_{t+1}}
+(a+b)\Bigl(\frac{1}{\delta_{t+1}}-\frac{1}{\Delta A_{t+1}}\Bigr) \ge 0,
\label{eq:case_ordering}
\end{equation}
The first term is the relaxation gap from the $c_{\max}$ bound, present already at the boundary $\delta_{t+1}=\Delta A_{t+1}$ and independent of $\bar D_{t+1}$. The second is the amortization gap arising when the energy requirement cap binds slack's dispatch below its full availability: it is positive whenever $\bar D_{t+1}<\sum_i A^{\alpha,\lambda}_{i,t+1}$ and grows as $\bar D_{t+1}$ decreases, so the Case 2 threshold is strictly decreasing in $\bar D_{t+1}$ on the interior. Operationally, a smaller next-event energy requirement $\bar D_{t+1}$ shrinks the volume $\delta_{t+1}$ over which the reallocation loss $C^{\mathrm{re}}_t$ can be repaid, forcing the next-event price $\pi_{t+1}$ to be correspondingly larger before fairness pays off.

Substantively, fairness pays off when the next event is jointly \emph{scarce enough} in price to clear the unit cost of the availability-mix change and \emph{large enough} in energy requirement to monetize the slack-induced volume. Because $A^0_{i,t+1}$ and $A^{\alpha,\lambda}_{i,t+1}$ depend on the realized shock $\sigma_t$ through the state transition $g(\cdot,\sigma_t)$, all threshold quantities--$C^{\mathrm{re}}_t$, $C^{\mathrm{mix}}_t$, $\Delta A_{t+1}$, $\delta_{t+1}$, $\bar c^{0}_{t+1}$, $c_{\max}$--are observable from the realized state and the consumer cost vector, so the operator can evaluate the condition event by event.

\section{Case Study}
We illustrate the proposed mechanism on a real-world case from Norway. Two empirical ingredients drive the operating model: consumer price-response behavior, identified from the iFlex demand-response dataset~\cite{hofmann2023rich}, and market conditions---prices and energy requirements---taken from the ENTSO-E day-ahead market record~\cite{ENTSO}. All the data and code are provided in~\cite{chen2026vpp}.

\subsection{Parameter estimation}

\paragraph{Heterogeneous response behavior}

We estimate consumer-specific costs $c_i$ under the linear specification via threshold identification. Under linear marginal cost, consumer $i$ is dispatched whenever $\pi_t>c_i$ and inactive otherwise, so $c_i$ is the price at which the dispatch probability transitions from low to high. For each household $i\in\mathcal N$ we fit a logistic regression of the dispatch indicator $\mathbf 1\{D_{i,t}>0\}$ on the contemporaneous price, $\mathbb P(D_{i,t}>0\mid\pi_t)=[1+\exp(-(a_i+b_i\pi_t))]^{-1}$, and extract $\hat c_i=-\hat a_i/\hat b_i$, the price at which the fitted probability crosses $0.5$. Households whose slope $\hat b_i$ is not significantly positive at the $5\%$ level are excluded for insufficient identification, leaving $257$ households. The resulting cost estimates exhibit substantial heterogeneity across consumers, as shown in Fig.~\ref{fig:parameter}.

\paragraph{State transition}

We specify $g(x)=1-\exp(-\eta x)$, whose increasing, concave, and saturating shape is consistent with the empirical response pattern in Fig.~\ref{fig:empirical}. We impose the EWMA normalization $\rho=1-\beta$, which interprets $S$ as a weighted average of past dispatch and reduces estimation to two stages. First, a household-specific AR(1) on event-to-event normalized dispatch, $r_{i,t}=\beta_i r_{i,t-1}+\mathrm{const}+\varepsilon_{i,t}$, yields $\hat\beta_i$ and $\hat\rho_i=1-\hat\beta_i$ for each of the $257$ consumers. Second, conditional on $(\hat\beta_i,\hat\rho_i)$, we reconstruct each household's latent state $S_{i,t}$ and fit $g$ to its realized state-availability pairs, producing household-specific curvature estimates $\hat\eta_i$.  To mitigate the downward bias that pooled estimation can introduce when aggregating heterogeneous concave response curves, we estimate $\eta$ at the household level and adopt a single representative value for the operating model. Restricting attention to households with regression fit $R^2\ge0.15$ (approximately 54 households) and taking the median yields $\hat\eta\approx 6$. Fig.~\ref{fig:parameter} reports the distributions of the estimated $\rho_i$ and $\eta_i$.

\paragraph{Market models}
The market primitives $(\pi_t,\bar D_t)$ are derived from the $T=225$ dispatch events in the February 2021 ENTSO-E day-ahead price and load series for bidding zone NO1~\cite{ENTSO}, matching the location and timing of the Norway field experiment. Because wholesale prices are substantially lower than the estimated participation thresholds ${\hat c_i}$, we rescale the price series by a factor of approximately $40$, placing it on a realistic VPP compensation scale while preserving its temporal variation and tail behavior. This ensures that the standing assumption $\pi_t>c_i$ holds throughout the horizon. The energy requirement $\bar D_t$ is likewise scaled to the household demand level of the iFlex experiment. Price and energy requirement are positively correlated, with a correlation coefficient of $0.58$. Fig.~\ref{fig:market} reports the resulting price and energy requirement series.

\begin{figure}[t]
    \centering
    \includegraphics[width=0.99\linewidth]{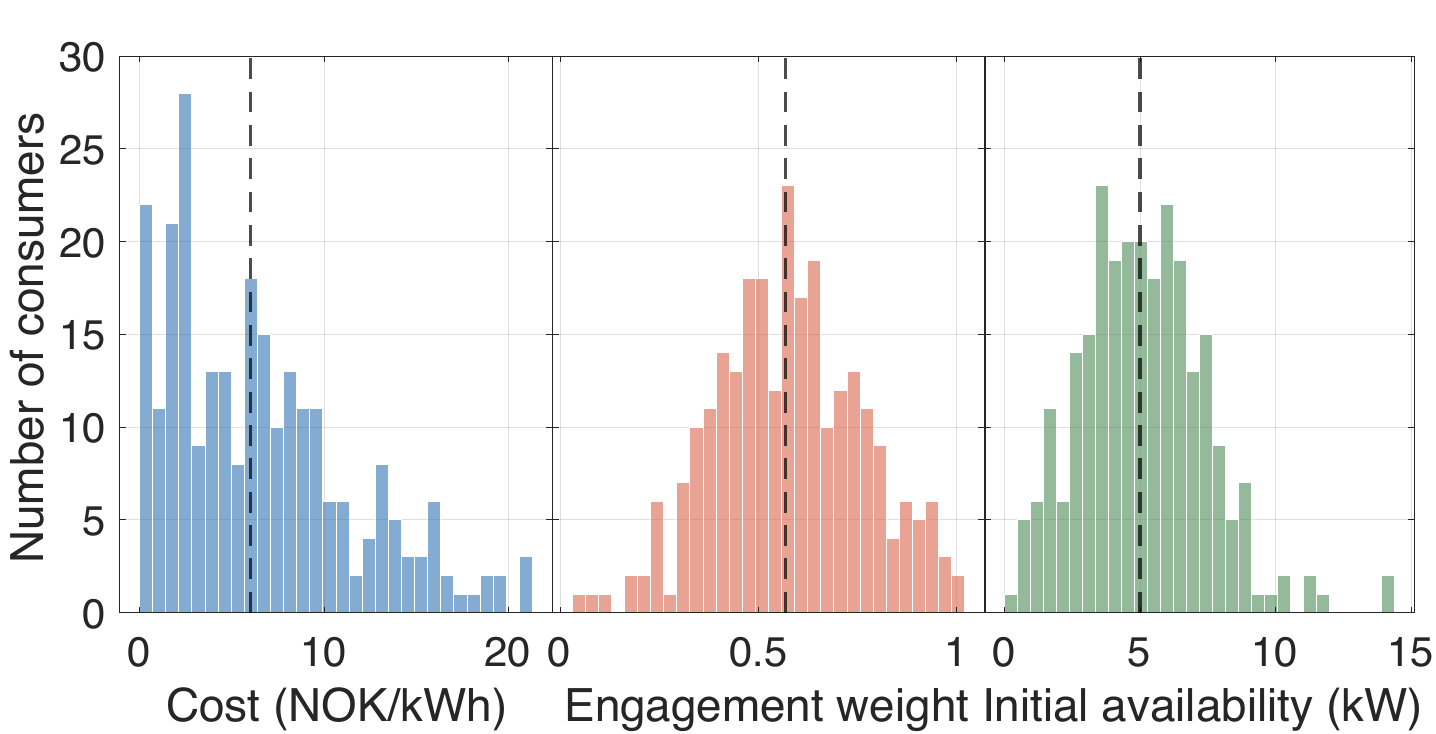}
    \caption{Estimated consumer cost parameters, engagement weight, and availability across the $257$ households.}
    \label{fig:parameter}
\end{figure}

\begin{figure}[t]
    \centering
    \includegraphics[width=0.99\linewidth]{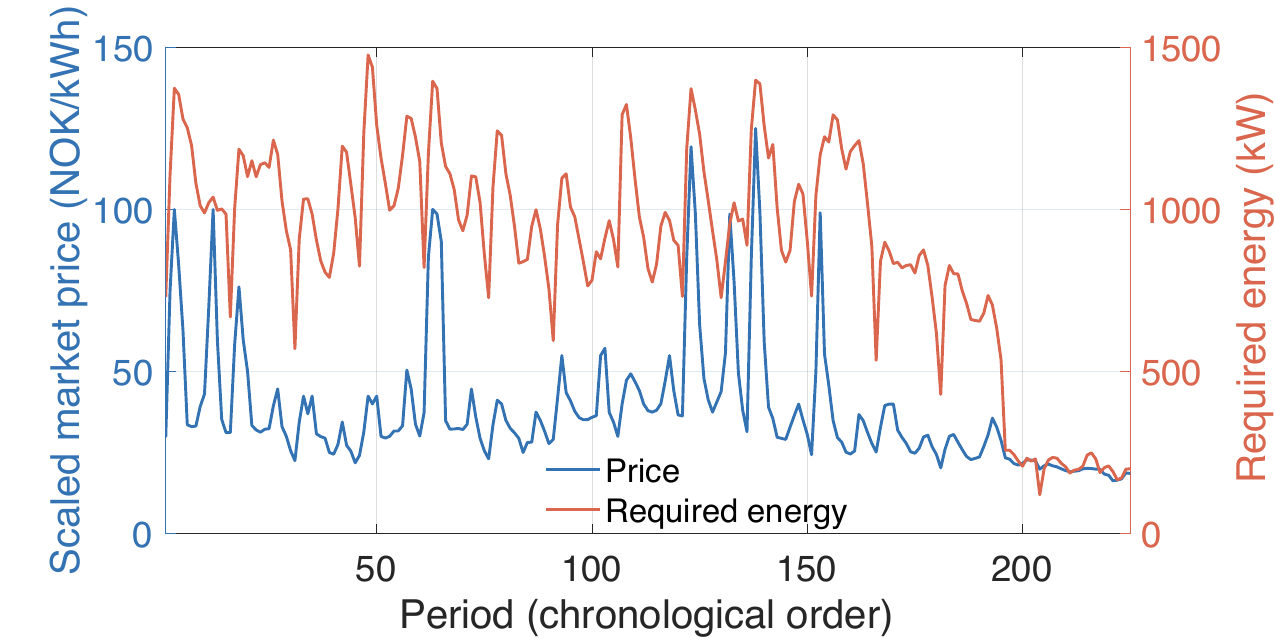}
    \caption{Scaled wholesale market price and required energy.}
    \label{fig:market}
\end{figure}

\begin{figure}[t]
    \centering
    \includegraphics[width=0.99\linewidth]{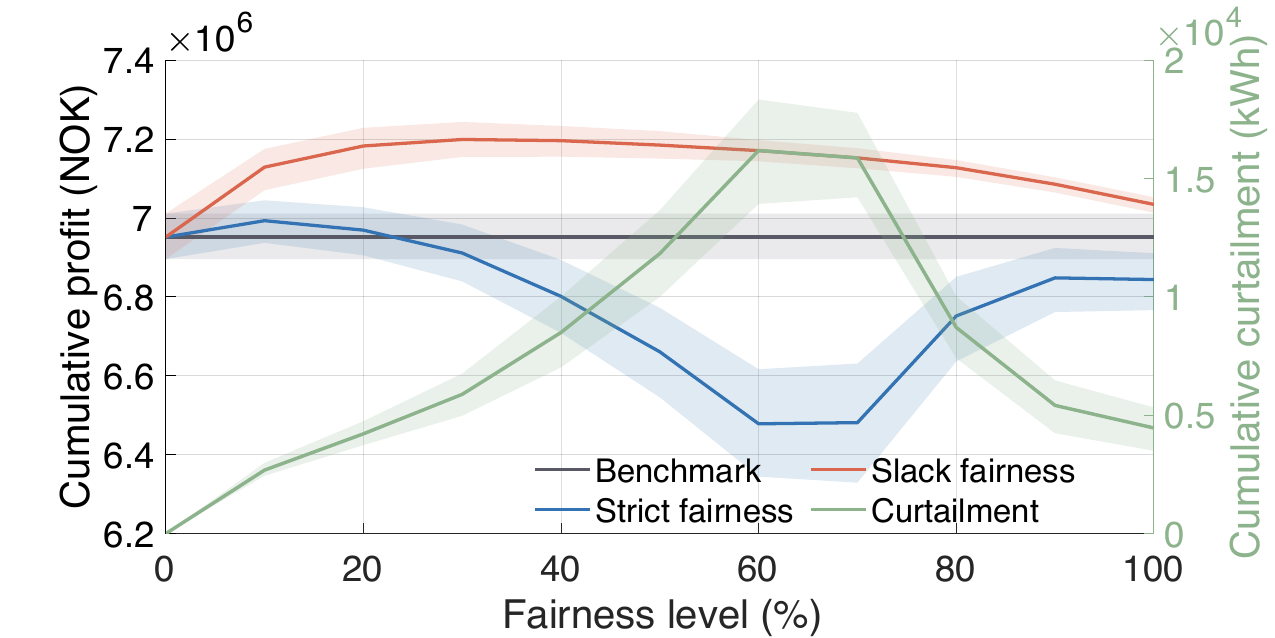}
    \caption{Cumulative aggregator profit (left axis) and curtailment (right axis)}.
    \label{fig:4profit}
\end{figure}

% \begin{figure}[t]
% \centering
% \begin{minipage}[t]{0.49\linewidth}
%   \centering
%   \includegraphics[width=\linewidth]{market.pdf}
%   \caption{Scaled wholesale market price and required energy.}
%   \label{fig:market}
% \end{minipage}\hfill
% \begin{minipage}[t]{0.49\linewidth}
%   \centering
%   \includegraphics[width=\linewidth]{4profit.pdf}
%   \caption{Cumulative aggregator profit and curtailment. Only curtailment uses the right y-axis.}
%   \label{fig:4profit}
% \end{minipage}
% \end{figure}

\subsection{Long-run profit, availability, and threshold validation}
We simulate the operating model on all $257$ consumers over $T=225$ periods. Energy provision and availability are normalized by each consumer's initial availability, and the fairness constraint is imposed on the normalized decision. To remove dependence on any single event order, all profit and curtailment curves are averaged over $100$ Monte-Carlo simulations bootstrapped from the ENTSO record, with $10$-$90\%$ ranges shown as shaded bands.

Fig.~\ref{fig:4profit} shows that strict fairness induces curtailment, particularly at moderate fairness levels where the range cap binds, consistent with Proposition~\ref{lem:curtail_cause}. The slack-augmented design eliminates this curtailment and recovers cumulative profit, as predicted by Theorem~\ref{thm:dominance}. Fig.~\ref{fig:ava} shows the aggregate availability trajectory for $\alpha=0.8$ in a representative sample. Consistent with Proposition~\ref{prop:availability_gain}, fairness increases aggregate availability relative to the benchmark, with the gains concentrated during high-price, high-energy requirement events where additional availability is most valuable. The availability paths under strict and slack fairness are nearly identical despite their profit difference, indicating that the curtailed service primarily comes from low-cost consumers whose availability is already near saturation on the concave transition function $g$.

To illustrate the availability dynamics at the individual-consumer level, Fig.~\ref{fig:individual_ava} reports four representative households selected to span the range of estimated participation costs: consumers 1 ($c_1=8.53$), 86 ($c_{86}=8.25$), 172 ($c_{172}=2.63$), and 257 ($c_{257}=7.41$) NOK/kWh. Under slack fairness with $\alpha=0.8$, where ``b'' denotes the benchmark and ``f'' denotes the fairness model, the higher-cost consumers (1 and 86) experience increased availability, while the lower-cost consumer (172) and the moderate-cost consumer (257) exhibit only modest reductions. This pattern is consistent with the mechanism in Proposition~\ref{prop:availability_gain}: fairness reallocates participation toward consumers with lower participation states, increasing their future availability. The sharp decline in availability after period $200$ is driven by the reduction in required energy rather than by the fairness mechanism itself.

Fig.~\ref{fig:threshold_validation} provides an event-level check of Theorem~\ref{thm:threshold_slack}. Each point plots the realized two-event profit gain of the slack-augmented fairness operation over the benchmark against the realized price margin relative to the theorem's threshold. Events to the right of zero satisfy the price condition and yield positive gains, while events below the threshold can be negative, consistent with the threshold characterization.

\begin{figure}[t]
    \centering
    \includegraphics[width=0.95\linewidth]{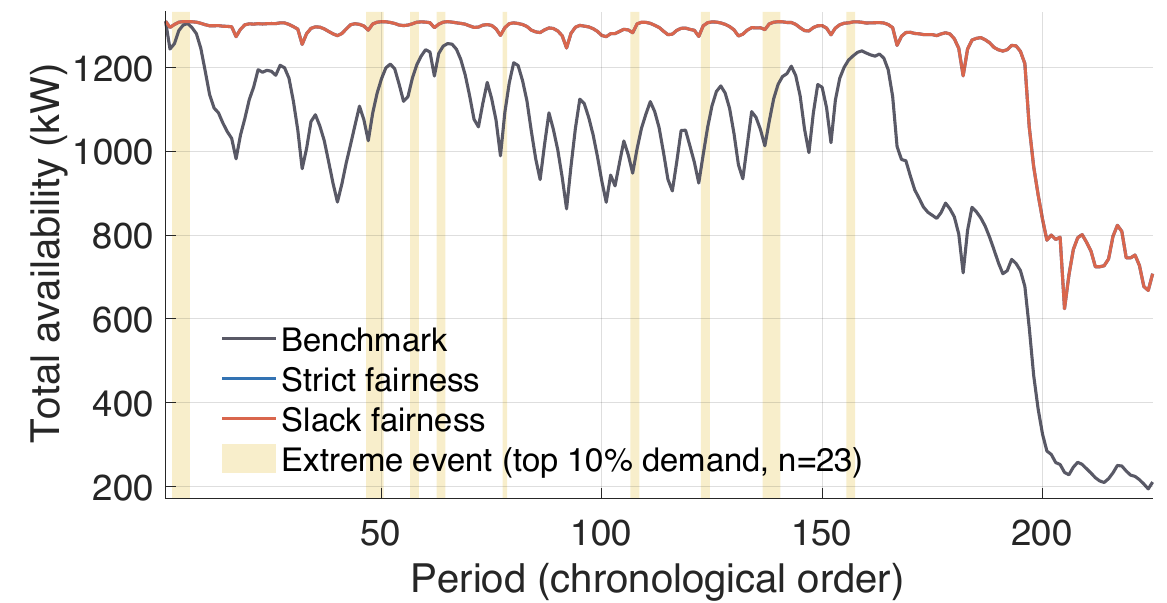}
    \caption{Aggregate consumer availability dynamics at $\alpha=0.8$.}
    \label{fig:ava}
\end{figure}

\begin{figure}[t]
    \centering
    \includegraphics[width=0.91\linewidth]{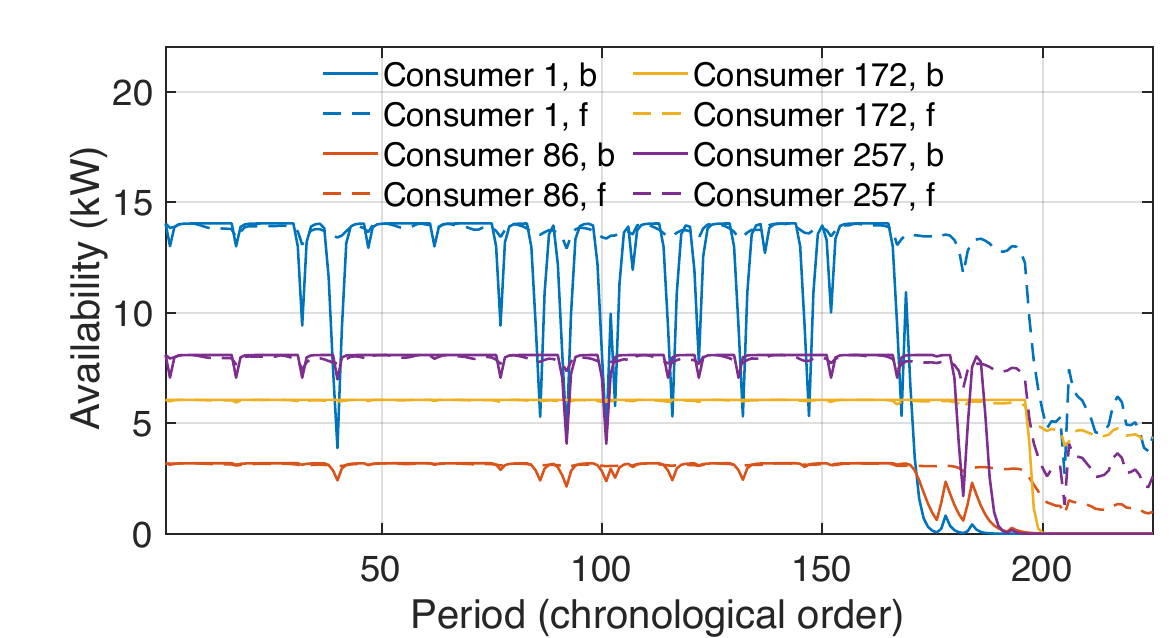}
    \caption{Exemplar consumer availability dynamics at $\alpha=0.8$.}
    \label{fig:individual_ava}
\end{figure}

\begin{figure}[t]
    \centering
    \includegraphics[width=0.95\linewidth]{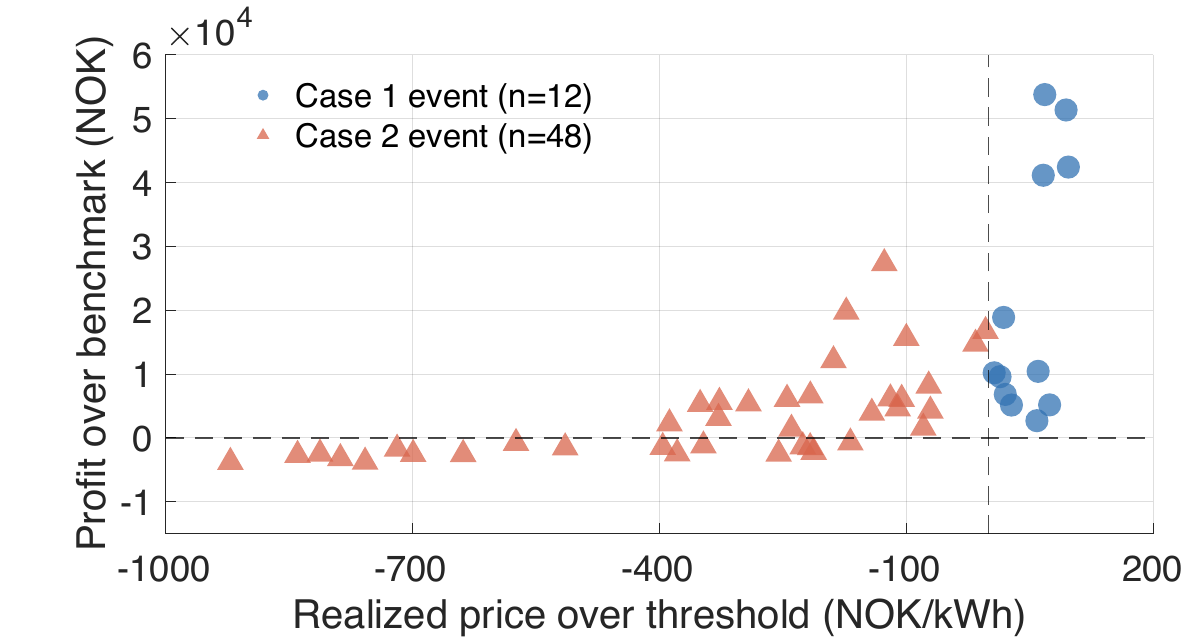}
    \caption{Event-level validation of the price threshold in Theorem~\ref{thm:threshold_slack}..}
    \label{fig:threshold_validation}
\end{figure}

% \begin{figure}[t]
% \centering
% \begin{minipage}[t]{0.49\linewidth}
%   \centering
%   \includegraphics[width=\linewidth]{ava.pdf}
%   \caption{Aggregate consumer availability dynamics at $\alpha=0.8$.}
%   \label{fig:ava}
% \end{minipage}\hfill
% \begin{minipage}[t]{0.49\linewidth}
%   \centering
%   \includegraphics[width=\linewidth]{threshold_validation.pdf}
%   \caption{Event-level validation of the price threshold in Theorem~\ref{thm:threshold_slack}.}
%   \label{fig:threshold_validation}
% \end{minipage}
% \end{figure}

\subsection{Trade-off between realized fairness and profit}
To measure the de facto fairness of realized dispatch from an external, third-party perspective, we adopt the Gini index~\cite{cowell2011measuring}
\begin{equation}
G(\alpha,\lambda):=\frac{\sum_{i\in\mathcal N}\sum_{j\in\mathcal N}|\tilde D_i^{\alpha,\lambda}-\tilde D_j^{\alpha,\lambda}|}{2N\sum_{i\in\mathcal N}\tilde D_i^{\alpha,\lambda}},
\label{eq:gini}
\end{equation}
where $\tilde D_i^{\alpha,\lambda}:=\sum_{t=1}^T D_{i,t}^{\alpha,\lambda}$ is consumer $i$'s cumulative dispatch, so $G=0$ is perfectly equal service and $G\to 1$ is full concentration on a single consumer.

\begin{figure}[t]
    \centering
    \includegraphics[width=0.95\linewidth]{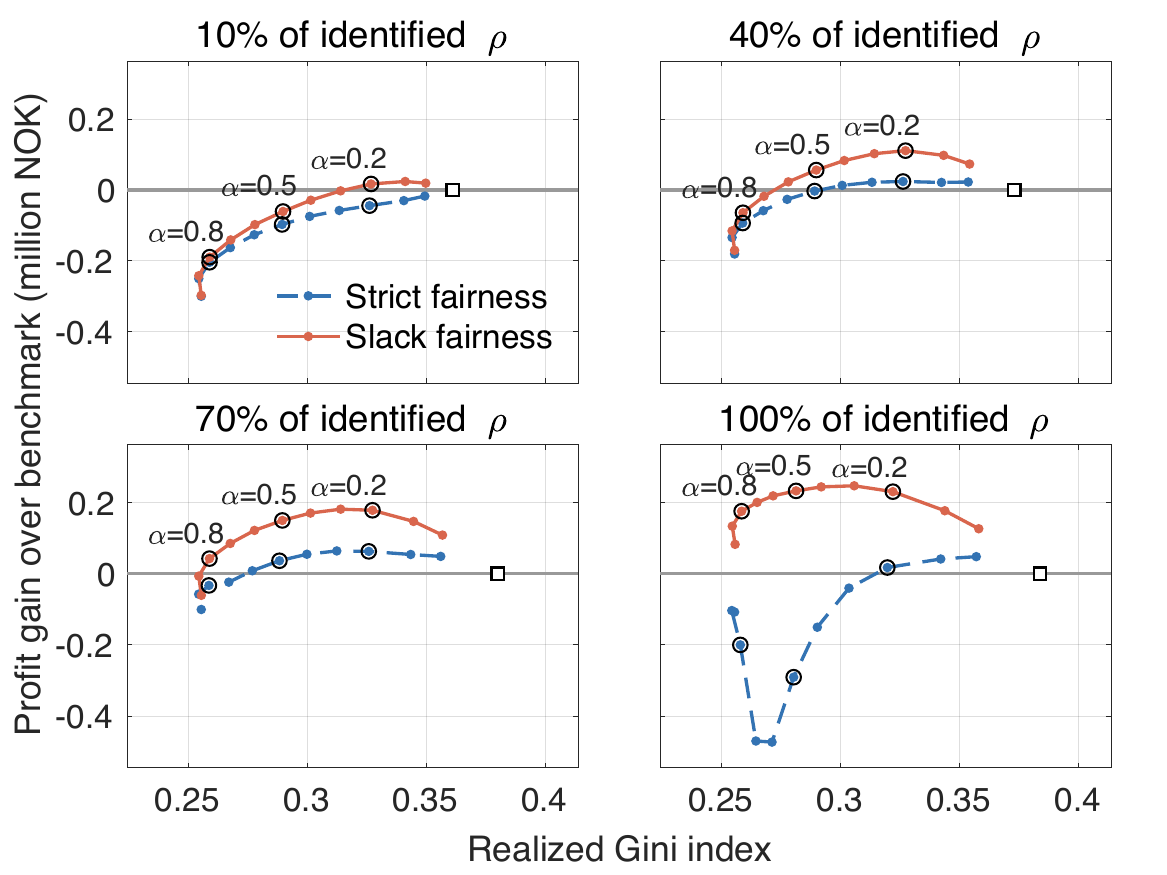}
    \caption{Pareto front of fairness-profit and its dependence on $\rho$. Squares mark the Gini index under the profit-only benchmark.}
    \label{fig:pareto}
\end{figure}

Fig.~\ref{fig:pareto} traces the realized Gini-profit trade-off. Raising $\alpha$ disperses dispatch and lowers the Gini index, from $0.36$-$0.38$ at the profit-only benchmark to about $0.25$ under near-equal dispatch. Because the intertemporal value of fairness rests entirely on the engagement weight $\rho$, we repeat the experiment with every consumer's identified $\hat\rho_i$ scaled to $10\%$, $40\%$, $70\%$, and $100\%$. Three findings emerge. First, the fairness-profit relation is strongly non-monotone: moderate fairness improves long-run profit, but aggressive equalization eventually erodes it by suppressing the most efficient consumers---under strict fairness, this can drive the profit gain negative. Second, slack fairness dominates strict fairness at every fairness level and every $\rho$ setting (Theorem~\ref{thm:dominance}), and the two are nearly vertically aligned in the figure: slack leaves the realized Gini index almost unchanged while delivering substantially higher profit, so the slack relaxation buys efficiency at almost no equity cost. Third, the profitability of fairness scales with $\rho$: as engagement weakens, the advantage shrinks, and at low $\rho$ even slack fairness becomes unprofitable at high $\alpha$, underscoring that intertemporal engagement is the channel through which fairness is monetized.

\section{Conclusion}

This paper examines fairness in VPP operations when consumer participation evolves endogenously in response to past dispatch decisions. We show that fairness has two distinct effects: it reallocates service across consumers, thereby increasing future aggregate availability under a concave participation dynamic, and it may induce curtailment, which always reduces profitability. This distinction reveals that fairness is not inherently costly. When participation reinforcement is present, a moderate degree of fairness can expand the aggregator's future flexibility and improve long-run performance.

To preserve the benefits of fairness while avoiding curtailment, we propose a slack-augmented fairness design that relaxes the fairness constraint only when necessary to maintain aggregate service. We show that this design weakly dominates strict fairness in realized profit and derive conditions under which the gains in availability from fairness outweigh the associated reallocation costs. Using consumer behavior and market data from Norway, we find that the fairness-profit relationship is non-monotone: moderate fairness can increase long-run profitability, whereas excessive equalization eventually reduces it.

More broadly, our results suggest that fairness should be viewed not only as an operational constraint or policy objective, but also as a mechanism for sustaining consumer engagement and expanding future flexibility. In practical VPP programs, participation is voluntary, and consumers are more likely to remain enrolled when they expect a reasonable opportunity to benefit from the program rather than observing rewards concentrated among a small subset of participants. From this perspective, fairness can serve as an investment in long-term participation by reinforcing the expectation that enrollment will lead to meaningful compensation opportunities over time. More generally, our findings suggest that fairness and profitability need not be opposing objectives; when participation responds to past outcomes, appropriately designed fairness mechanisms can strengthen both.

\bibliographystyle{ieeetr}
\bibliography{Ref}

\appendices

\section{Proof of Proposition~\ref{prop:availability_gain}}
\begin{proof}
Define the shock-averaged transition $f(S):=\mathbb E_{\sigma}[g(S,\sigma)]$. Since expectation preserves monotonicity and concavity, $f$ is strictly increasing and strictly concave whenever each $g(\cdot,\sigma)$ is.

Fix $t$ and condition on $\mathcal F_t$. Under linear costs the benchmark objective $\sum_{i\in \mathcal N}(\pi_t-c_i)D_{i,t}$ has a unique optimizer $D_t^0$ since $c_1<\cdots<c_N$; let $D_t^\alpha$ be any fairness optimizer. We use the transfer characterization of majorization~\cite{majorization}: $x\succ y$ iff $y$ can be obtained from $x$ via finitely many Pigou--Dalton transfers.

\emph{Step 1: ordering.}
We claim that under either benchmark or fairness, for every $\tau\le t$,
\begin{equation}
    D_{1,\tau}\ge \cdots \ge D_{N,\tau},
\quad
S_{1,\tau+1}\ge \cdots \ge S_{N,\tau+1}.
\end{equation}
At $\tau=1$, all availabilities are equal. If some feasible allocation satisfies $D_i<D_j$ for $i<j$, shifting $\varepsilon>0$ from $j$ to $i$ preserves feasibility (and the fairness range constraint, if present) but increases the objective by $(c_j-c_i)\varepsilon>0$, contradicting optimality. Hence $D_{1,1}\ge\cdots\ge D_{N,1}$, and the update~\eqref{eq:oracle_dynamics} yields the same order for $S_{i,2}$.

Assume $S_{1,\tau}\ge\cdots\ge S_{N,\tau}$. Since availability is ordered accordingly, the same exchange argument gives $D_{1,\tau}\ge\cdots\ge D_{N,\tau}$, and~\eqref{eq:oracle_dynamics} preserves the order of $S_{\tau+1}$, proving the claim.

\emph{Step 2: $D_t^0\succ D_t^\alpha$.}
Since $\alpha\Delta_t>0$, the fairness cap is strictly tighter, so $D_t^\alpha\neq D_t^0$. By assumption, $\sum_i D_{i,t}^\alpha=\sum_i D_{i,t}^0$. Because both vectors are ordered as in Step 1, define
\begin{equation}
    \mathcal N^+:=\{i:\ D_{i,t}^0>D_{i,t}^\alpha\},\quad
\mathcal N^-:=\{i:\ D_{i,t}^0<D_{i,t}^\alpha\}.
\end{equation}
Both sets are nonempty, and total surplus equals total deficit. Starting from $x^{(0)}:=D_t^0$, repeatedly transfer
\begin{equation}
  \delta:=\min\{x^{(0)}_i-D_{i,t}^\alpha,\ D_{j,t}^\alpha-x^{(0)}_j\}>0, 
\end{equation}
from any $i\in\mathcal N^+$ to any $j\in\mathcal N^-$ still away from target. Since donor coordinates are larger than receiver coordinates, each step is a Pigou--Dalton transfer; the $L_1$ distance to $D_t^\alpha$ strictly decreases, so the procedure terminates finitely at $D_t^\alpha$. Hence $D_t^0\succ D_t^\alpha$.

\emph{Step 3: strict concavity increases total expected availability.}
Let $x^{(0)}=D_t^0,\ x^{(1)},\dots,x^{(L)}=D_t^\alpha$ be the transfer sequence from Step 2, where $x^{(\ell+1)}=x^{(\ell)}-\delta_\ell e_{p_\ell}+\delta_\ell e_{q_\ell}$ with $x^{(\ell)}_{p_\ell}\ge x^{(\ell)}_{q_\ell}$; for a generic step we write the index pair $(p,q)$ and size $\sigma$, and define
\begin{equation}
   u=\beta S_{p,t}+\rho x_{p},\quad v=\beta S_{q,t}+\rho x_{q}. 
\end{equation}
By Step 1, $u\ge v$. After the transfer, the two arguments become $u-\rho\delta$ and $v+\rho\delta$. Since $f$ is strictly concave,
\begin{equation}
    f(u-\rho\delta)+f(v+\rho\delta)>f(u)+f(v).
\end{equation}
All other terms are unchanged, so each transfer strictly increases $\sum_i f(S_{i,t+1})$, and summing gives $\sum_{i\in \mathcal N} f(S_{i,t+1}^\alpha)>\sum_{i\in \mathcal N} f(S_{i,t+1}^0)$. Since $f(S_{i,t+1})=\mathbb E[A_{i,t+1}\mid\mathcal F_t]$, this is $\sum_{i\in \mathcal N}\mathbb E[A_{i,t+1}^\alpha\mid \mathcal F_t] > \sum_{i\in \mathcal N}\mathbb E[A_{i,t+1}^0\mid \mathcal F_t]$. If $f$ is affine, equality holds at every step, so the totals are equal.
\end{proof}

\section{Proof of Theorem~\ref{thm:dominance}}
\begin{proof}
Fix period $t$. By~\eqref{con:nocurtail}, the slack-augmented optimizer satisfies $\sum_{i\in\mathcal N}D_{i,t}^{\alpha,\lambda}=\sum_{i\in\mathcal N}D_{i,t}^0$. By definition, the strict fairness optimizer satisfies $\sum_{i\in\mathcal N}D_{i,t}^\alpha=\sum_{i\in\mathcal N}D_{i,t}^0-R_t(\alpha)$. Subtracting,
\begin{equation}
\sum_{i\in\mathcal N}D_{i,t}^{\alpha,\lambda}-\sum_{i\in\mathcal N}D_{i,t}^\alpha=R_t(\alpha)\ge 0.
\label{eq:agg_diff}
\end{equation}

\emph{Case 1: $R_t(\alpha)=0$.} At $s=0$ the shared box and aggregate constraints~\eqref{eq:oracle_box}--\eqref{eq:oracle_aggregate}, together with the slack-augmented constraints~\eqref{con:fairness_slack}--\eqref{con:nocurtail}, reduce to the strict fairness constraints: \eqref{con:fairness_slack} collapses to~\eqref{con:fairness}, and~\eqref{con:nocurtail} is implied by $R_t(\alpha)=0$. Since $\lambda\ge\mathcal C_1$, positive slack cannot improve the penalized objective unless it is required to satisfy~\eqref{con:nocurtail}. Thus $(D_t^\alpha,0)$ is feasible for the slack-augmented problem and any feasible $(D,0)$ is feasible for strict fairness, so the two share an optimizer and $\Pi_t^{\alpha,\lambda}=\Pi_t^\alpha$.

\emph{Case 2: $R_t(\alpha)>0$.} We invoke Proposition~\ref{prop:curtail_cost} with $D_t:=D_t^{\alpha,\lambda}$, $\tilde D_t:=D_t^\alpha$, and curtailment profile $r_i:=D_{i,t}^{\alpha,\lambda}-D_{i,t}^\alpha$. By~\eqref{eq:agg_diff}, $\sum_i r_i=R_t(\alpha)>0$; it remains to show $r_i\ge 0$ pointwise. Index consumers by increasing cost, $c_1\le\cdots\le c_N$. By the exchange argument of Step~1 in the proof of Proposition~\ref{prop:availability_gain}, every fairness optimizer is sorted, $D_{1,t}\ge\cdots\ge D_{N,t}$, and under linear costs takes the bang-bang form: a top block saturated at the availability cap $A_{i,t}$, a middle block at the common upper level $\overline D_t$, and a bottom block at the floor $\underline D_t:=\overline D_t-(1-\alpha)\Delta_t$ set by the binding range cap. Both $D_t^\alpha$ and $D_t^{\alpha,\lambda}$ share this sorted form; comparing block by block, the top block is pinned at the cap $A_{i,t}$ in both, so $r_i=0$ there. The slack solution relaxes the range cap to $(1-\alpha)\Delta_t+s_t^{\alpha,\lambda}$, so its upper level satisfies $\overline D_t^{\alpha,\lambda}=\overline D_t^\alpha+s_t^{\alpha,\lambda}\ge\overline D_t^\alpha$, giving $r_i\ge 0$ on the middle block. The remaining mass $R_t(\alpha)=\sum_i r_i>0$ is absorbed by the bottom block, which by the sorted ordering can only raise the floor, $\underline D_t^{\alpha,\lambda}\ge\underline D_t^\alpha$, so $r_i\ge 0$ there too. Hence $r_i\ge 0$ for all $i$, and Proposition~\ref{prop:curtail_cost} Part~(1) gives $\Pi_t^\alpha<\Pi_t^{\alpha,\lambda}$.

Summing~\eqref{eq:dominance} over $t$ gives the horizon-level claim, with strict inequality whenever Case 2 occurs.
\end{proof}

\section{Proof of Theorem~\ref{thm:threshold_slack}}
\begin{proof}
    We decompose the event-pair profit gap into an event-$t$ and an event-$(t{+}1)$ contribution.

\emph{Step 1: event-$t$ contribution.} The slack-augmented LP at $t$ enforces the no-curtailment constraint~\eqref{con:nocurtail}, $\sum_i D^{\alpha,\lambda}_{i,t}=\sum_i D^0_{i,t}$, so
\begin{align}
&\Pi^{\alpha,\lambda}_t-\Pi^0_t
=\sum_i(\pi_t-c_i)\bigl(D^{\alpha,\lambda}_{i,t}-D^0_{i,t}\bigr)\notag\\
&=\pi_t \sum_i\bigl(D^{\alpha,\lambda}_{i,t}-D^0_{i,t}\bigr) - \sum_i c_i\bigl(D^{\alpha,\lambda}_{i,t}-D^0_{i,t}\bigr) = -C^{\mathrm{re}}_t.\label{eq:step_t}
\end{align}

\emph{Step 2: benchmark at $t{+}1$ under scarcity.} Under either case condition, $\bar D_{t+1}>\sum_i A^0_{i,t+1}$ (in Case 1, from $\bar D_{t+1}\ge\sum_i A^{\alpha,\lambda}_{i,t+1}>\sum_i A^0_{i,t+1}$ since $\Delta A_{t+1}>0$; in Case 2, stated directly), so the aggregate energy requirement cap $\sum_i D_{i,t+1}\le\bar D_{t+1}$ does not bind at the benchmark optimum. Since $\pi_{t+1}>c_i$ for every $i$, the per-event profit $\sum_i(\pi_{t+1}-c_i)D_{i,t+1}$ is strictly increasing in each $D_{i,t+1}$, so $D^0_{i,t+1}=A^0_{i,t+1}$ for all $i$, yielding
\begin{equation}
\Pi^0_{t+1} = \sum_i(\pi_{t+1}-c_i)A^0_{i,t+1} = \Bigl(\sum_i A^0_{i,t+1}\Bigr)\bigl(\pi_{t+1}-\bar c^{0}_{t+1}\bigr).\label{eq:bench_tplus1}
\end{equation}

\emph{Step 3 event-$(t{+}1)$ contribution.}

\textbf{Case 1}: With $\bar D_{t+1}\ge \sum_i A^{\alpha,\lambda}_{i,t+1}$, the aggregate cap is non-binding at the slack-augmented LP as well, and the same argument gives $D^{\alpha,\lambda}_{i,t+1}=A^{\alpha,\lambda}_{i,t+1}$ for all $i$; the no-curtailment constraint at $t{+}1$ is automatically satisfied at the slack-state benchmark, and the fairness range cap is inactive. Hence
\begin{equation}
\begin{aligned}
    \Pi^{\alpha,\lambda}_{t+1}-\Pi^0_{t+1}
    &=\sum_i(\pi_{t+1}-c_i)\bigl(A^{\alpha,\lambda}_{i,t+1}-A^0_{i,t+1}\bigr)\\
    &=\pi_{t+1} \Delta A_{t+1}-C^{\mathrm{mix}}_t.
\end{aligned}
\label{eq:gap_tplus1_case1}
\end{equation}
Combining~\eqref{eq:step_t} and~\eqref{eq:gap_tplus1_case1},
\begin{equation}
    \bigl(\Pi^{\alpha,\lambda}_t+\Pi^{\alpha,\lambda}_{t+1}\bigr)-\bigl(\Pi^0_t+\Pi^0_{t+1}\bigr) =-C^{\mathrm{re}}_t + \pi_{t+1}\Delta A_{t+1}-C^{\mathrm{mix}}_t.
\end{equation}
Since $g(\cdot,\sigma_t)$ is strictly concave for the realized shock, the Pigou-Dalton argument of Proposition~\ref{prop:availability_gain} applies pathwise, giving $\Delta A_{t+1}>0$ for the realized $\sigma_t$; hence this is non-negative if and only if $\pi_{t+1}\ge(C^{\mathrm{re}}_t+C^{\mathrm{mix}}_t)/\Delta A_{t+1}$, establishing~\eqref{eq:threshold_case1}.

\textbf{Case 2}: Under $\sum_i A^0_{i,t+1}<\bar D_{t+1}<\sum_i A^{\alpha,\lambda}_{i,t+1}$, the slack-augmented LP at $t{+}1$ faces a binding aggregate cap. Its no-curtailment constraint forces $\sum_i D^{\alpha,\lambda}_{i,t+1} = \bar D_{t+1}$. The LP optimum cannot be written in closed form without specifying which consumers are selected by the merit ordering at the slack state; we instead bound it below by a feasible candidate dispatch that depends only on aggregates.

From $\pi_{t+1}>c_i$ for every $i\in\mathcal N$ and $\sum_i D^{\alpha,\lambda}_{i,t+1}=\bar D_{t+1}$,
\begin{align}
\Pi^{\alpha,\lambda}_{t+1}
&= \sum_i(\pi_{t+1}-c_i)D^{\alpha,\lambda}_{i,t+1} \nonumber \\
 &\ge (\pi_{t+1}-c_{\max})\sum_i D^{\alpha,\lambda}_{i,t+1}
 = (\pi_{t+1}-c_{\max})\bar D_{t+1},
\label{eq:slack_lp_lb}
\end{align}
% where $c_{\max}:=\max_i c_i$.

Combining~\eqref{eq:bench_tplus1} and~\eqref{eq:slack_lp_lb},
\begin{equation}
    \begin{aligned}
&\Pi^{\alpha,\lambda}_{t+1}-\Pi^0_{t+1}
\ge\bar D_{t+1}\bigl(\pi_{t+1}-c_{\max}\bigr)-\Bigl(\sum_i A^0_{i,t+1}\Bigr)\bigl(\pi_{t+1}-\bar c^{0}_{t+1}\bigr)\\
&=\pi_{t+1} \delta_{t+1} - \Bigl[\bar D_{t+1}\,c_{\max}-\bigl(\textstyle\sum_i A^0_{i,t+1}\bigr)\bar c^{0}_{t+1}\Bigr].\label{eq:gap_tplus1_case2}
\end{aligned}
\end{equation}
Further combining it with~\eqref{eq:step_t}, we have a sufficient condition for non-negativity
\begin{equation}
    \pi_{t+1} \delta_{t+1}\ge C^{\mathrm{re}}_t + \bar D_{t+1}\,c_{\max} - \bigl(\textstyle\sum_i A^0_{i,t+1}\bigr)\bar c^{0}_{t+1}.
\end{equation}
Writing $\bar D_{t+1}=\sum_i A^0_{i,t+1}+\delta_{t+1}$ and rearranging, dividing through by $\delta_{t+1}>0$ obtain the~\eqref{eq:threshold_case2}.
\end{proof}

\end{document}